\documentclass[journal]{IEEEtran}
\usepackage{multirow}%added by yxw 20161022
\usepackage{diagbox}%added by yxw 20161022
\usepackage{amssymb}
\usepackage{amsmath}
\usepackage{graphicx}
\usepackage{cite}
\usepackage{citesort}
\usepackage{subfigure}
\usepackage{graphicx,epstopdf}
\usepackage{epsfig}	
\usepackage{cite,graphicx,amsmath,amssymb}
\usepackage{comment}

\usepackage{amssymb}
\usepackage{amsmath}
\usepackage{cite}
\usepackage{url}
\usepackage{xcolor}
\usepackage{cite,graphicx,amsmath,amssymb}
\usepackage{subfigure}
\usepackage{citesort}
\usepackage{fancyhdr}
\usepackage{mdwmath}
\usepackage{mdwtab}
\usepackage{caption}
\usepackage{amsthm}
\usepackage{lipsum}

\newtheorem{theorem}{Theorem}

\newtheorem{lemma}{Lemma}

\newtheorem{corollary}{Corollary}

\newtheorem{remark}{Remark}  % added by yxw 20160929

\makeatletter
\def\ScaleIfNeeded{%
\ifdim\Gin@nat@width>\linewidth \linewidth \else \Gin@nat@width
\fi } \makeatother

\begin{document}

\title{\Huge{Performance Analysis of NOMA with Fixed Gain Relaying over Nakagami-$m$ Fading Channels}}

\author{ Xinwei~Yue,~\IEEEmembership{Student Member,~IEEE,} Yuanwei\ Liu,~\IEEEmembership{Member,~IEEE,} Shaoli~Kang, Arumugam~Nallanathan,~\IEEEmembership{Fellow,~IEEE}

\thanks{X. Yue is with School of Electronic and Information Engineering, Beihang university, Beijing,
China (email: xinwei$\_$yue@buaa.edu.cn).}
\thanks{Y. Liu and A. Nallanathan are with King's College London, London,
U.K. (email: \{yuanwei.liu, arumugam.nallanathan\}@kcl.ac.uk).}
\thanks{S. Kang is with State Key Laboratory of Wireless Mobile Communications, China Academy of Telecommunications Technology (CATT), Beijing, China and also with School of Electronic and Information Engineering, Beihang University, Beijing, China (email: kangshaoli@catt.cn).}
 }

%\author{
%\IEEEauthorblockN{  Xinwei~Yue\IEEEauthorrefmark{1}, Yuanwei~Liu\IEEEauthorrefmark{2}, Shaoli Kang\IEEEauthorrefmark{3},  %Arumugam~Nallanathan,\IEEEauthorrefmark{2},}
%\IEEEauthorblockA{
%\IEEEauthorrefmark{1} Beihang University, Beijing, China\\
%\IEEEauthorrefmark{2} King's College London, London, UK\\
%\IEEEauthorrefmark{3} China Academy of Telecommunication Technology, Beijing, China\\
%\IEEEauthorrefmark{4}  Beijing University of Posts and Telecommunications, Beijing, China\\ %School of Information and Communication Engineering,
%\IEEEauthorrefmark{4}  Henan Polytechnic University, China.\\%school of Physics and Electronic Information Engineering,
% } }
%\IEEEauthorrefmark{3} King's College London, London, UK

%\thanks{X. Yue is with School of Electronic and Information Engineering, Beihang university, Beijing
%China (email: xinwei$\_$yue@buaa.edu.cn).}
%\thanks{Y. Liu and A. Nallanathan are with King's College London, London,
%UK (email: \{yuanwei.liu, arumugam.nallanathan\}@kcl.ac.uk).}
%\thanks{S. Kang is with China Academy of Telecommunication Technology, Beijing, China (email: kangshaoli@catt.cn).}
%\thanks{ Z. Ding is with School of Computing and Communications, Lancaster University, LA1 4WA,
% UK (e-mail: z.ding@lancaster.ac.uk).}
%}

%Arumugam~Nallanathan

\maketitle

\begin{abstract}
This paper studies the application of cooperative techniques for non-orthogonal multiple access (NOMA). More particularly, the fixed gain amplify-and-forward (AF) relaying with NOMA is investigated over Nakagami-$m$ fading channels. Two scenarios are considered insightfully. 1) The first scenario is that the base station (BS) intends to communicate with multiple users through the assistance of AF relaying, where the direct links are existent between the BS and users; and 2) The second scenario is that the AF relaying is inexistent between the BS and users. To characterize the performance of the considered scenarios, new closed-form expressions for both exact and asymptomatic outage probabilities are derived.
Based on the analytical results, the diversity orders achieved by the users are obtained. For the first and second scenarios, the diversity order for the $n$-th user are $\mu(n+1)$ and $\mu n$, respectively. Simulation results unveil that NOMA is capable of outperforming orthogonal multiple access (OMA) in terms of outage probability and system throughput. It is also worth noting that NOMA can provide better fairness compared to conventional OMA.
By comparing the two scenarios, cooperative NOMA scenario can provide better outage performance relative to the second scenario.%\textcolor[rgb]{0.00,0.00,1.00}{}
\end{abstract}
\begin{keywords}
Non-orthogonal multiple access, amplify-and-forward relaying, Nakagami-$m$ fading channels, diversity order
\end{keywords}
%keywords: Non-orthogonal multiple access (NOMA); Amplify-and-Forward relaying; Nakagami-$m$ fading channels; Outage probability; Diversity order;
\section{Introduction}

Non-orthogonal multiple access (NOMA) has received considerable attention as the promising technique for future wireless networks due to its superior spectral efficiency and massive connectivity \cite{Li20145G,6666209}. The pivotal feature of NOMA is that signals from the plurality of users can share and multiplex the same radio resources with different power factors based on their channel conditions. At the receiving end, the user with poor channel conditions regards other user's messages as interference when it decodes its own message. However, the user with better channel conditions is capable of getting rid of another users' messages by applying successive interference cancellation (SIC) before decoding its own information \cite{Cover1991Elements}.

Some initial research contributions in the field of NOMA have been made by researchers \cite{Ding2015Application,6868214,7361990,7273963,Liu7398134}. More specifically, in \cite{Ding2015Application}, the authors summarized the emerging technologies from NOMA combination with multiple-input multiple-output (MIMO) to cooperative NOMA and cognitive radio (CR) NOMA, etc. In the cellular down link scenario, the outage behavior of NOMA with randomly deployed users was investigated using bounded path loss model in \cite{6868214}. In \cite{7361990}, the authors derived the outage probability under two different kinds of channel state information (CSI). The influence of user pairing with the fixed power allocation for NOMA system over Rayleigh fading channels was analyzed in \cite{7273963}. Furthermore, the performance of NOMA in large-scale underlay CR was evaluated in terms of outage probability by using stochastic-geometry \cite{Liu7398134}. To evaluate the performance of uplink NOMA, the outage probability of more efficient NOMA schemes with power control has been derived in \cite{7390209}. The authors investigated the mutil-cell uplink NOMA transmission scenarios using Poisson cluster process \cite{Tabassum2016uplink}, in which the rate coverage probability for the NOMA user was derived on the conditions of the different SIC schemes.

Wireless relaying technology, which has given rise to the extensive attention, is an effective way to combat the deleterious effects of fading. The outage performance of the amplify-and-forward (AF) and decode-and-forward (DF) relaying schemes were investigated in \cite{laneman2004cooperative}. Recently, several contributions in term of NOMA with relaying have been researched \cite{Ding2014Cooperative,6708131,7445146,7230246,Choi7561047}, which can improve the spectrum efficiency and transmit reliability of wireless network. Cooperative NOMA scheme was first proposed in which users with better channel conditions are delegated as relaying nodes \cite{Ding2014Cooperative}. As such, the communication reliability for the users far away from the base station are enhanced. The coordinated two-point system with superposition coding (SC) was investigated in the down link communication \cite{6708131}. In order to improve energy efficiency, the authors have considered the simultaneous wireless information and power transfer (SWIPT) to NOMA \cite{7445146}, in which stochastic geometry has been utilized to model the positions of users and the near user is regarded as a energy harvesting DF relay to help far user.
The outage probability and achievable average rate of NOMA with DF relaying were analyzed over Rayleigh fading channels \cite{7230246}. The author in \cite{Choi7561047} proposed relay-aided multiple access (RASA), in which
the near user exploiting the two way relaying protocol to help far user. Additionally, for solving the potential time slot wasted brought by half-duplex relaying protocol, the outage performance of full-duplex device-to-device based cooperative NOMA was researched in \cite{Ding2016FD}.
\subsection{Motivation and Contributions}
While the aforementioned literature laid a solid foundation for the role of NOMA in Rayleigh fading, the impact of cooperative NOMA in Nakagami-$m$ fading has not been well understood.
%The Nakagami-$m$ distribution can be reduce to a variety of fading distributions by setting different value of the parameter $m$. Such as, the Gaussian distribution $(m=\frac{1}{2})$ and the Reyleigh distribution $(m=1)$ are special cases of its.
Based on the different parameter settings, Nakagami-$m$ fading channel can be reduce to multiple types of channel. For instance, the Gaussian channel $(\mu=\frac{1}{2})$ and Rayleigh fading channel $(\mu=1)$ are the special cases of its. In \cite{Ding7740931Nakagami}, the authors investigated the spectrum-sharing CR network based cooperative NOMA over Nakagami-$m$ fading channels. The outage performance of NOMA with channel sorted referring the variable gain AF relaying have been researched in \cite{Men7454773}, but the impact of NOMA in terms of the direct link transmission in Nakagami-$m$ fading was not considered. Therefore, the prior work in \cite{Men7454773} motivates us to develop this research contribution.

%\textcolor[rgb]{0.00,0.00,1.00}{\footnote{\textcolor[rgb]{0.00,0.00,1.00}{It is assumed that these two scenarios are evaluated in terms of outage probability respectively, but our future
%work will compare the performance of these scenarios.}}}

To the best of our knowledge, the performance of NOMA with sorted channel referring to the BS with fixed gain AF relaying over Nakagami-$m$ fading channels is not researched yet. Additionally, as stated in \cite{6868214}, the author did not investigate the outage performance of downlink NOMA system over Nakagami-$m$ fading channels. Motivated by these, we address two NOMA transmission scenarios in this paper: 1) The first scenario is that the BS intends to communicate with multiple users through the assistance of AF relaying, where the direct links are existent between the BS and users; and 2) AF relaying is not existent between the BS and users. The primary contributions of this paper are summarised as follows:
\begin{enumerate}
  \item  We first derive the closed-form expressions of outage probability for the sorted NOMA users. To obtain more insights, we further derive the asymptotic outage probability of the users and obtain the corresponding diversity orders. We demonstrate that NOMA is capable of outperforming OMA in terms of outage probability over Nakagami-$m$ fading channels. We observe that when several users' quality of service (QoS) are met at the same time, NOMA can offer better fairness.
  \item  Additionally, we analyze the delay-limited transmission throughput for both scenarios based on the analytical results. It is worth noting that NOMA can achieve larger throughput with regard to conventional MA in more general channels. %We observe that when several users' QoS are met at same time, NOMA scheme can offer better fairness.
\end{enumerate}
%\begin{enumerate}
%  \item \emph{The first scenario:} We first derive the closed-form expressions of outage probability for the sorted NOMA user. To obtain more insights, we further derive the asymptotic outage probability of the users and obtain the corresponding diversity orders. We demonstrate that NOMA is capable of outperforming OMA in terms of outage probability over Nakagami fading channels. In addition, we analyze the delay-limited transmission throughput based on the derived outage probability.
 % \item \emph{The second scenario:} Similar to the first scenario, we derive the closed-form expressions for both exact and asymptomatic outage probabilities as well as delay-limited throughput. Moreover, the diversity order achieved for the sorted user is obtained. We observe that when several users' QoS are met at same time, NOMA scheme can offer better fairness with regard to conventional MA in more general channels.
%\end{enumerate}
%We first derive the closed-form expressions of outage probability for the near user and far user in two scenarios. Based on the analytical results, the diversity orders for the users are obtained. We demonstrate that NOMA is capable of outperforming OMA in terms of direct both scenarios. Additionally, we observe that when several users' QoS are met at same time, NOMA scheme can offer better fairness with regard to conventional MA in more general channels.
 \subsection{Organization}
The rest of this paper is organized as follows. Section \uppercase\expandafter{\romannumeral2} describes the system model for studying NOMA with the fixed gain AF relaying over Nakagami-$m$ fading channels. In Section \uppercase\expandafter{\romannumeral3}, the exact and asymptomatic expressions of outage probability for the users are derived in two scenarios. Numerical results are presented in Section \uppercase\expandafter{\romannumeral4} for verifying our analysis, and are followed by our conclusion in Section \uppercase\expandafter{\romannumeral5}.

\section{System Model}\label{System Model}
%This paper considers a downlink single cell cooperative communication system with a fixed gain AF relaying.
%\textcolor[rgb]{0.00,0.00,1.00}{}
This paper considers two insightful scenarios which are the downlink single cell cooperative communication scenario with a fixed gain AF relaying and non-cooperative communication scenario, respectively. For the sake of simplicity, the BS, a AF relaying node and two paired users which include near user ${d_n}$ and far user $d_{f}$
are presented as shown in Fig. 1, where the relaying node can be existent or inexistent. All the nodes are equipped with single antenna. The complex channel gain between the BS and users, between the BS and AF relaying node, and between the AF relaying node and users are denoted as ${h_{sd}}$, ${h_{sr}}$ and ${h_{rd}}$, respectively. Without loss of generality, the channel gains of $M$ users are sorted as ${\left| {{h_{s{d_1}}}} \right|^2} \le {\left| {{h_{s{d_2}}}} \right|^2} \le  \cdots  \le {\left| {{h_{s{d_M}}}} \right|^2}$\footnote{In this paper, we only focus our attention on investigating a sorted pair of users in which user 1 and user 2 can be selected or user 1 and user 3
are selected for performing NOMA jointly in the first scenario.}. All the complex channel gains are modeled as independent and identically distribution (i.i.d) random variables RVs $x$ which is subject to Nakagami-$m$ distribution \cite{1576942}.
\begin{figure}[t!]             %   According to the principle of NOMA transmission
\centering
\includegraphics[width=3 in]{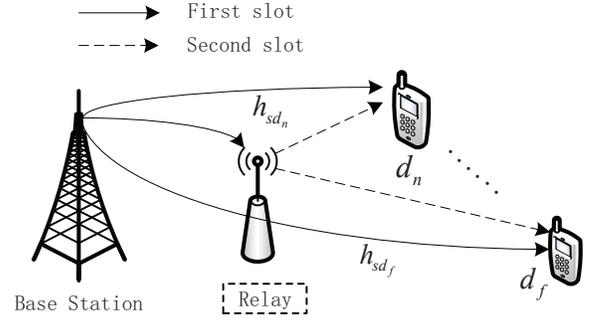}
 \caption{System model}
\label{Fig. 1}
\end{figure}
The transmission powers for the BS and the AF relaying node are assumed to be equal, i.e., $\left( {{P_s} = {P_r} = P} \right)$. The energy of the transmitted signal is normalized to one. Meanwhile, the additive white Gaussian noise (AWGN) terms of all the links have zero mean and variance ${N_0}$.
\subsection{The First Scenario}
For the first scenario, the whole communication processes are completed in two slots.
%In the first slot, the BS broadcasts the linear superposition of the signals of multiple users. The signals are received by the AF relaying node, ${d_n}$ and ${d_f}$ at the same time. In the second slot, the AF relaying node amplifies and forwards its received signals to ${d_n}$ and ${d_f}$ based on the AF protocol \cite{laneman2004cooperative}.
During the first slot, the BS transmits superposed signal $\sqrt {{a _n}{P_s}} {x_n} + \sqrt {{a _f}{P_s}} {x_f}$ to the relaying node, ${d_n}$ and ${d_f}$ according to the NOMA scheme \cite{6868214}. ${a_n}$ and ${a_f}$ are the power allocation coefficients for ${d_n}$ and ${d_f}$, where ${a _n} + {a _f} = 1$, ${a _f} > {a _n}$.  ${x_n}$ and ${x_f}$ are the signal for ${d_{n}}$ and ${d_{f}}$, respectively. By stipulating this assumption, SIC can be invoked by ${d_n}$ for first detecting ${d_f}$ having a larger transmit power, which has less inference signal. Accordingly, the signal of ${d_f}$ is detected from original superposed signal. The observation at the relaying node, ${d_{n}}$ and ${d_{f}}$ are given by
\begin{align}\label{express1 }
{y_r} = {h_{sr}}\left( {\sqrt {{a _n}{P_s}} {x_n} + \sqrt {{a _f}{P_s}} {x_f}} \right) + {n_{sr}},
\end{align}
\begin{align}\label{express2 }
{y_{{d_n}}} = {h_{s{d_n}}}\left( {\sqrt {{a _n}{P_s}} {x_n} + \sqrt {{a _f}{P_s}} {x_f}} \right) + {n_{s{d_n}}},
\end{align}
\begin{align}\label{express3 }
{y_{{d_f}}} = {h_{s{d_f}}}\left( {\sqrt {{a _n}{P_s}} {x_n} + \sqrt {{a _f}{P_s}} {x_f}} \right) + {n_{s{d_f}}},
\end{align}
where ${n_{sr}}$, ${n_{s{d_n}}}$ and ${n_{s{d_f}}}$ are AWGN at the relaying node, ${d_n}$ and ${d_f}$, respectively. The received signal to interference and noise ratio (SINR) for ${d_f}$ to detect ${x_n}$ is given by
\begin{align}\label{far SINR}
{\gamma _{s{d_f}}} = \frac{{{{\left| {{h_{s{d_f}}}} \right|}^2}{a _f}\rho }}{{{{\left| {{h_{s{d_f}}}} \right|}^2}{a _n}\rho  + 1}},
\end{align}
where ${\rho}=\frac{P_{s}}{N_{0}}$ is transmit signal to noise ratio (SNR). SIC is first performed for ${d_n}$ by detecting and decoding the $d_{f}$' information. Then, the received SINR at ${d_n}$ is given by
\begin{align}\label{near decod far SINR}
{\gamma _{s{d_{f \to n}}}} = \frac{{{{\left| {{h_{s{d_n}}}} \right|}^2}{a _f}\rho }}{{{{\left| {{h_{s{d_n}}}} \right|}^2}{a _n}\rho  + 1}}.
\end{align}
After the far user’ message is decoded, ${d_n}$ can decode its own information with the following SINR
\begin{align}\label{near SINR}
{\gamma _{s{d_n}}} = {\left| {{h_{s{d_n}}}} \right|^2}{a _n}\rho.
\end{align}

During the second slot, the relaying node amplifies the received signal and forwards to ${d_n}$ and ${d_f}$ using the fixed gain factor $\kappa  = \sqrt {\frac{{{P_r}}}{{{P_s}\mathbb{E}\left( {{{\left| {{h_{sr}}} \right|}^2}} \right) + {N_0}}}}$, where $\mathbb{E}\{\cdot\}$ denotes expectation operation. The signals received at ${d_n}$ and ${d_f}$ is expressed as
\begin{align}\label{recevied signal at dn}
 {y_{r{d_n}}} =& \kappa {h_{r{d_n}}}{h_{sr}}\left( {\sqrt {{a _n}{P_s}} {x_n} + \sqrt {{a _f}{P_s}} {x_f}} \right)\nonumber\\
  &+ \kappa {h_{r{d_n}}}{n_{sr}} + {n_{r{d_n}}}
\end{align}
and
\begin{align}\label{recevied signal at df }
 {y_{r{d_f}}} =& \kappa {h_{r{d_f}}}{h_{sr}}\left( {\sqrt {{a _n}{P_s}} {x_n} + \sqrt {{a _f}{P_s}} {x_f}} \right) \nonumber\\
  &+ \kappa {h_{r{d_f}}}{n_{sr}} + {n_{r{d_f}}}
\end{align}
respectively, where ${n_{r{d_n}}}$ and ${n_{r{d_n}}}$ denote the AWGN at ${d_n}$ and ${d_f}$, respectively. The received SINR for ${d_f}$ to detect $x_{f}$ is given by
\begin{align}\label{SINR at df second}
{\gamma _{r{d_f}}} = \frac{{{{\left| {{h_{sr}}} \right|}^2}{{\left| {{h_{r{d_f}}}} \right|}^2}{a _f}\rho }}{{{{\left| {{h_{sr}}} \right|}^2}{{\left| {{h_{r{d_f}}}} \right|}^2}{a _n}\rho  + {{\left| {{h_{r{d_f}}}} \right|}^2} + C}},
\end{align}
where $C \buildrel \Delta \over = 1/{{\rm{\kappa }}^2}$. ${d_n}$ first detect $d_{f}$'s information with the received SINR given by
\begin{align}\label{near user detect far user SINR at second}
{\gamma _{r{d_{f \to n}}}} = \frac{{{{\left| {{h_{sr}}} \right|}^2}{{\left| {{h_{r{d_n}}}} \right|}^2}{a _f}\rho }}{{{{\left| {{h_{sr}}} \right|}^2}{{\left| {{h_{r{d_n}}}} \right|}^2}{a _n}\rho  + {{\left| {{h_{r{d_n}}}} \right|}^2} + C}},
\end{align}
and then after SIC operations, the receiving SINR for ${d_n}$ is given by
\begin{align}\label{SINR at dn second}
{\gamma _{r{d_n}}} = \frac{{{{\left| {{h_{sr}}} \right|}^2}{{\left| {{h_{r{d_n}}}} \right|}^2}{a _n}\rho }}{{{{\left| {{h_{r{d_n}}}} \right|}^2} + C}}.
\end{align}
\subsection{The Second Scenario}
 On the basis of the above scenario, another scenario considered in this paper is that the AF relaying node is assumed to be absent with randomly user deployment.

 For the second scenario, the BS transmits the superposed signals to all the users based on the NOMA scheme. Therefore, the signal received at the $m$-th user is written as
\begin{align}\label{SINR at dn second }
{y_m} = {h_m}\mathop \sum \limits_{j = 1}^M \sqrt {{a _j}{P_s}} {x_j} + {n_m},
\end{align}
where ${h_m}$ denotes the Nakagami-$m$ fading channel gain from the BS to the $m$-th user. ${a_j}$ is the power allocation coefficient for the $j$-th user with $\mathop \sum \limits_{j = 1}^M {a _j} = 1$, while it satisfies the relationship for
${a_1} \ge {a_2} \ge  \cdots  \ge {a_{\rm{M}}}$. $x_j$ denotes the signal for the $j$-th user and ${n_m}$ is AWGN at the $m$-th user. Thus, SIC is employed at the $m$-th user and the receiving SINR for the $m$-th user to detect the $i$-th user $\left( {1 \le i \le m \le M} \right)$ is given by
\begin{align}\label{near detect far SINR at second secor  }
{\gamma _{i \to m}} = \frac{{{{\left| {{h_m}} \right|}^2}{a _i}\rho }}{{\rho {{\left| {{h_m}} \right|}^2}\sum\limits_{j = i + 1}^M {{a _j}}  + 1}}.
\end{align}

After $M-1$ users can be detected successfully, the received SINR for the $M$-th user is given by
\begin{align}\label{near SINR at second secor  }
{\gamma _M} = {\left| {{h_M}} \right|^2}{a _M}\rho.
\end{align}

\section{Performance evaluation}\label{Performance evaluation}
In this section, the performance of two scenarios are characterized in terms of outage probability as follows.
\subsection{Outage Probability}
 It is significant to examine the outage probability when the user’ QoS requirements can be satisfied in the communication system just as in [6]. The outage probability of the users over Nakagami-$m$ fading channels is analyzed for two different scenarios.

 From the above explanations, the probability density function (PDF) for $x=|h|$ is expressed as
\begin{align}\label{h for pdf}
f\left( x \right) = \frac{{2{\mu ^\mu }}}{{{\rm{\Gamma }}\left( \mu  \right)\omega _0^\mu }}{x^{2\mu  - 1}}{e^{ - \frac{{\mu {x^2}}}{{{\omega _0}}}}},x > 0
\end{align}
where ${\rm{\Gamma }}\left(  \cdot  \right)$ is the Gamma function, $\mu$ and ${\omega _0}$ denote the parameters of the multipath fading and the control spread, respectively. Therefore, $\lambda  = {x^2}$ is subject to the Gamma distribution. The PDF and cumulative distribution function (CDF) of $\lambda $ is expressed as \cite{Simon2005Digital}
\begin{align}
f\left( \lambda  \right) = \frac{{{\mu ^\mu }{\lambda ^{\mu  - 1}}}}{{\omega _0^\mu {\rm{\Gamma }}\left( \mu  \right)}}{e^{ - \frac{{\mu \lambda }}{{{\omega _0}}}}},\lambda  \ge 0
\end{align}
and
\begin{align} \label{CDF for unorder}
F\left( \lambda  \right) = 1 - {e^{ - \frac{{\mu \lambda }}{{{\omega _0}}}}}\mathop \sum \limits_{k = 0}^{\mu  - 1} \frac{1}{{k!}}{\left( {\frac{{\mu \lambda }}{{{\omega _0}}}} \right)^k}, \lambda  \ge 0
\end{align}
respectively, where ${\omega _0} = {\rm{\mathbb{E}}}\left[ {{\lambda ^2}} \right]$ is the average power.%${\rm{\mathbb{E}}}\left[  \cdot  \right]$ denotes the statistical expectation of a RV,

With the aid of order statistics \cite{David2003Order} and binomial theorem, the PDF and CDF of the $m$th user’s channel gain ${\left| {{h_m}} \right|^2}$ can be expressed as
\begin{align}
 {f_{{{\left| {{h_m}} \right|}^2}}}\left( x \right) = &\frac{{M!}}{{\left( {m - 1} \right)!\left( {M - m} \right)!}}{f_{{{\left| h \right|}^2}}}\left( x \right)\nonumber\\
  &\times {\left( {{F_{{{\left| h \right|}^2}}}\left( x \right)} \right)^{m - 1}}{\left( {1 - {F_{{{\left| h \right|}^2}}}\left( x \right)} \right)^{M - m}},
\end{align}
and
\begin{align}\label{CDF for order}
 {F_{{{\left| {{h_m}} \right|}^2}}}\left( x \right) = \frac{{M!}}{{\left( {m - 1} \right)!\left( {M - m} \right)!}}\mathop \sum \limits_{i = 0}^{M - m} {M-m \choose
  i   } \nonumber\\
  \times \frac{{{{\left( { - 1} \right)}^i}}}{{m + i}}{\left( {{F_{{{\left| h \right|}^2}}}\left( x \right)} \right)^{m + i}},\nonumber\\
\end{align}
respectively, where ${\left| h \right|^2}$ is the unsorted channel gain between the BS and an arbitrary user.
%$\choose $
\subsubsection{Outage Probability for the First Scenario}\label{Outage Probability with the firstScenario}
In this scenario, the users combine with the observations from the BS and the relaying node by using selection combining at the last slot. Therefore, an outage event for ${d_f}$ can be interpreted as two reasons, i.e., it cannot detect its own message at both slots. Based on the above explanation, the outage probability of ${d_f}$ is given by
\begin{align} \label{expression OP for df}
{P_{{d_f}}} =& {{\rm{P}}_{\rm{r}}}\left( {{\gamma _{s{d_f}}} < {\gamma _{t{h_f}}}} \right){{\rm{P}}_{\rm{r}}}\left( {{\gamma _{r{d_f}}} < {\gamma _{t{h_f}}}} \right),
\end{align}
where ${\gamma _{t{h_f}}} = {2^{2{R_f}}} - 1$ with ${R_f}$ being the target rate at ${d_f}$.

The following theorem provides the outage probability of $d_{f}$ in the this scenario.
\begin{theorem} \label{theorem:1}
The closed-form expression for the outage probability of the investigated $d_{f}$ is expressed as
\begin{align} \label{the last expression OP for df}
 {P_{{d_f}}} =& \mathop \sum \limits_{i = 0}^{M - f} {
 M-f \choose
  i}\frac{{{\varphi _f}}}{{f + i}}\mathop \sum \limits_{q = 0}^{f + i} {
 f+i \choose
  q}{\left( { - 1} \right)^{q + i}}{\chi _f} \nonumber\\
  &\times \mathop \sum \limits_{{p_0} +  \cdots  + {p_{\mu  - 1}} = q} {q \choose
  {{p_0}, \cdots ,{p_{\mu  - 1}}}  } \mathop \prod \limits_{k = 0}^{\mu  - 1} {\left( {\frac{{\psi _f^k}}{{k!}}} \right)^{{p_k}}}  \nonumber\\
  &\times \left\{ {1 - \frac{{2{\mu ^\mu }{e^{ - \frac{{\mu \varepsilon }}{{{\omega _{sr}}}}}}}}{{\omega _{sr}^\mu {\rm{\Gamma }}\left( \mu  \right)}}\mathop \sum \limits_{k = 0}^{\mu  - 1} \frac{{{{\left( {\varepsilon C} \right)}^k}}}{{k!}}{{\left( {\frac{\mu }{{{\omega _{r{d_f}}}}}} \right)}^k}\mathop \sum \limits_{i = 0}^{\mu  - 1} {\mu-1 \choose
  i }} \right.  \nonumber\\
& \left. { \times {\varepsilon ^{\mu  - i - 1}}{{\left( {\frac{{\varepsilon C{\omega _{sr}}}}{{{\omega _{r{d_f}}}}}} \right)}^{\frac{{i - k + 1}}{2}}}{K_{i - k + 1}}\left( {2\mu \sqrt {\frac{{\varepsilon C}}{{{\omega _{sr}}{\omega _{r{d_f}}}}}} } \right)} \right\},
\end{align}
where $\varepsilon  \buildrel \Delta \over = \frac{{{\gamma _{t{h_f}}}}}{{\rho \left( {{a _f} - {a _n}{\gamma _{t{h_f}}}} \right)}}$ with ${a _f} > {a _n}{\gamma _{t{h_f}}}$. $M$ denotes the number of users in the considered scenario, ${\varphi _f} = \frac{{M!}}{{\left( {f - 1} \right)!\left( {M - f} \right)!}}$, ${\chi _f} = {e^{ - q{\psi _f}}}$, ${\psi _f} = \frac{{\mu \varepsilon }}{{{\omega _{s{d_f}}}}}$, ${q \choose
  {{p_0}, \cdots ,{p_{\mu  - 1}}}  } = \frac{{q!}}{{{p_0}!{p_1}! \cdots {p_{\mu  - 1}}!}}$. $f$ denotes the $f$-th user (far~user). ${\omega _{sr}}$ and ${\omega _{r{d_f}}}$ denote the average power for the links between the BS and the relaying node, and between the relaying node and $d_f$, respectively. ${K_v}\left(  \cdot  \right)$ is the modified Bessel function of the second kind with order $v$.
\end{theorem}
\begin{proof} See Appendix A.
\end{proof}

According to NOMA scheme, the outage would not occur for $d_n$ in two situations where $d_n$ can detect $d_f$'s information and also can detect its own information during the two slots. Furthermore, the outage probability of $d_n$ is given by
\begin{align} \label{expression OP for dn}
 {P_{{d_n}}} =& \left[ {1 - {{\rm{P}}_{\rm{r}}}\left( {{\gamma _{s{d_{f \to n}}}} > {\gamma _{t{h_f}}},{\gamma _{s{d_n}}} > {\gamma _{t{h_n}}}} \right)} \right] \nonumber \\
 & \times \left[ {1 - {{\rm{P}}_{\rm{r}}}\left( {{\gamma _{r{d_{f \to n}}}} > {\gamma _{t{h_f}}},{\gamma _{r{d_n}}} > {\gamma _{t{h_n}}}} \right)} \right],
\end{align}
where ${\gamma _{t{h_n}}} = {2^{2{R_n}}} - 1$ with ${R_n}$ being the target rate at ${d_n}$.

The following theorem provides the outage probability of $d_{n}$ in this scenario.
\begin{theorem} \label{theorem:2}
The closed-form expression for the outage probability of the investigated $d_{n}$ is expressed as
%\newpage
\begin{align} \label{the last expression OP for dn}
 {P_{{d_n}}} =& \mathop \sum \limits_{i = 0}^{M - n} {m-n \choose
  i  }\frac{{{\varphi _n}}}{{n + i}}\mathop \sum \limits_{q = 0}^{n + i} {n+i \choose
  q  }{\left( { - 1} \right)^{q + i}}{\chi _n} \nonumber\\
 & \times \mathop \sum \limits_{{p_0} +  \cdots  + {p_{\mu  - 1}} = q}{q \choose
  {{p_0}, \cdots ,{p_{\mu  - 1}}}   }\mathop \prod \limits_{k = 0}^{\mu  - 1} {\left( {\frac{{\psi _n^k}}{{k!}}} \right)^{{p_k}}}\nonumber\\
&  \times \left\{ {1 - \frac{{2{\mu ^\mu }{e^{ - \frac{{\mu {\rm{\Omega }}}}{{{\omega _{sr}}}}}}}}{{\omega _{sr}^\mu {\rm{\Gamma }}\left( \mu  \right)}}\mathop \sum \limits_{k = 0}^{\mu  - 1} \frac{{{{\left( {{\rm{\Omega }}C} \right)}^k}}}{{k!}}{{\left( {\frac{\mu }{{{\omega _{r{d_n}}}}}} \right)}^k}\mathop \sum \limits_{i = 0}^{\mu  - 1} {\mu-1 \choose
  i   }} \right.\nonumber \\
 &\left. { \times {\varepsilon ^{\mu  - i - 1}}{{\left( {\frac{{{\rm{\Omega }}C{\omega _{sr}}}}{{{\omega _{r{d_n}}}}}} \right)}^{\frac{{i - k + 1}}{2}}}{K_{i - k + 1}}\left( {2\mu \sqrt {\frac{{{\rm{\Omega }}C}}{{{\omega _{sr}}{\omega _{r{d_n}}}}}} } \right)} \right\},
\end{align}
where $\beta  \buildrel \Delta \over = \frac{{{\gamma _{t{h_n}}}}}{{{a _n}\rho }}$, ${\rm{\Omega }} \buildrel \Delta \over = \max \left( {\varepsilon ,\beta } \right)$, ${\varphi _n} = \frac{{M!}}{{\left( {n - 1} \right)!\left( {M - n} \right)!}}$, ${\chi _n} = {e^{ - q{\psi _n}}}$, ${\psi _n} = \frac{{\mu {\rm{\Omega }}}}{{{\omega _{s{d_n}}}}}$, ${\omega _{r{d_n}}}$ denotes the average power of the link between the relaying and $d_n$. $n$ denotes the $n$-th user (near~user).
\end{theorem}
\begin{proof} See Appendix B.
\end{proof}
\subsubsection{Outage Probability for the Second Scenario}\label{Outage Probability with the second Scenario}
In this scenario, the SIC is carried out at the $m$-th user by detecting and canceling the $i$-th user’s information $\left( {i \le m} \right)$ before it detects and decodes its own signals in terms of NOMA protocol. If the $m$-th user cannot detect the discretionary $i$-th user’s information, outage occurs. Therefore, after some manipulations such as in [6], the outage probability of $m$-th user can be expressed as follows:
\begin{align} \label{expression for second scenario}
{P_m} = {{\rm{P}}_{\rm{r}}}\left( {{{\left| {{h_m}} \right|}^2} < \varphi _m^*} \right),
\end{align}
where $\varphi _m^* = \max \left\{ {{\varphi _1},{\varphi _2}, \cdots ,{\varphi _m}} \right\}$, ${\varphi _i} = \frac{{{\gamma _{t{h_i}}}}}{{\rho ({a _i} - {\gamma _{t{h_i}}}\mathop \sum \limits_{j = i + 1}^M {a _j})}}$ for $i < M$, ${\varphi _M} = \frac{{{\gamma _{t{h_M}}}}}{{\rho {a _M}}}$, ${\gamma _{t{h_i}}} = {2^{{R_i}}} - 1$ with ${R_i}$ being the target rate at $i$th user. Note that \eqref{expression for second scenario} is obtained under the condition of ${a_i} > {\gamma _{t{h_i}}}\sum\nolimits_{j = i + 1}^M {{a_j}} $.

Substituting \eqref{CDF for unorder} into \eqref{CDF for order}, the outage probability of the $m$-th user over Nakagami-$m$ fading channels can be given by
\begin{align} \label{op derived for second}
 {P_m} =& \frac{{M!}}{{\left( {m - 1} \right)!\left( {M - m} \right)!}}\mathop \sum \limits_{i = 0}^{M - m} {
 M-m \choose
  i}\frac{{{{\left( { - 1} \right)}^i}}}{{m + i}} \nonumber\\
  &\times \mathop \sum \limits_{q = 0}^{m + i} {
 m+i \choose
  q}{\left( { - 1} \right)^q}{e^{ - \frac{{\mu \varphi _m^*q}}{{{\omega _{\rm{m}}}}}}} \nonumber\\
  &\times \mathop \sum \limits_{{p_0} +  \cdots  + {p_{\mu  - 1}} = q} {q \choose
  {{p_0}, \cdots ,{p_{\mu  - 1}}} }\mathop \prod \limits_{k = 0}^{\mu  - 1} {\left( {\frac{{\psi _m^k}}{{k!}}} \right)^{{p_k}}}, \nonumber\\
\end{align}
where ${\psi _m} = \frac{{\mu \varphi _m^*}}{{{\omega _{{m}}}}}$. ${\omega _{{m}}}$ is the average power of the link between the BS and the $m$-th user.

\subsection{Diversity Analysis}\label{Diversity Analysis}
In this section, to gain more insights, the diversity order achieved by the users for two scenarios can be obtained based on the above analytical results. The diversity order is defined as
\begin{align}\label{diversity order}
d =  - \mathop {\lim }\limits_{\rho  \to \infty } \frac{{\log \left( {P \left( \rho  \right)} \right)}}{{\log \rho }}.
\end{align}

When $\varepsilon\to 0$, the approximate expressions of CDF for the unsorted channel gain ${\left| h \right|^2}$ and the $f$-th
user's sorted channel gain ${\left| {{h_f}} \right|^2}$ are given by \cite{Men7454773}
\begin{align}\label{unsorted channel gain CDF}
{F_{{{\left| h \right|}^2}}}\left( \varepsilon  \right) \approx {\left( {\frac{{\mu \varepsilon }}{{{\omega _{sr}}}}} \right)^\mu }\left( {\frac{1}{{\mu !}}} \right),
\end{align}
and
\begin{align}\label{sorted channel gain CDF}
{F_{{{\left| {{h_f}} \right|}^2}}}\left( \varepsilon  \right) \approx \frac{{M!}}{{\left( {M - f} \right)!f!}}{\left( {\frac{{\mu \varepsilon }}{{{\omega _0}}}} \right)^{\mu f}}{\left( {\frac{1}{{\mu !}}} \right)^f},
\end{align}
respectively.

Define the two probabilities at the right hand side of \eqref{expression OP for df} by ${{\rm{\Theta }}_1}$ and ${{\rm{\Theta }}_2}$ respectively. Based on \eqref{sorted channel gain CDF}, a high SNR approximation $(\varepsilon  \to 0)$ of ${{\rm{\Theta }}_1}$ is given by
\begin{align}\label{appro1 op}
{{\rm{\Theta }}_1} \approx \frac{{M!}}{{\left( {M - f} \right)!f!}}{\left( {\frac{{\mu \varepsilon }}{{{\omega _{s{d_f}}}}}} \right)^{\mu f}}{\left( {\frac{1}{{\mu !}}} \right)^f} \propto \frac{1}{{{\rho ^{\mu f}}}},
\end{align}
where $ \propto $ represents ``be proportional to".

${\Theta _2}$ can be rewritten as follows:
\begin{align}
{\Theta _2}{\rm{ = }}{{\rm{P}}_{\rm{r}}}\left( {{{\left| {{h_{sr}}} \right|}^2} < \varepsilon } \right) + \int_\varepsilon ^\infty  {{f_{{{\left| {{h_{sr}}} \right|}^2}}}\left( y \right)} {F_{{{\left| {{h_{r{d_f}}}} \right|}^2}}}\left( {\frac{{\varepsilon C}}{{y - \varepsilon }}} \right)dy.
\end{align}

With the aid of \eqref{unsorted channel gain CDF} and \eqref{sorted channel gain CDF}, the approximation expression of ${\Theta _2}$ at high SNR is given by
\begin{align}\label{appro2 op}
{\Theta _2} \approx {\left( {\frac{{\mu \varepsilon }}{{{\omega _{sr}}}}} \right)^\mu }\left( {\frac{1}{{\mu !}}} \right) + {\left( {\frac{{\mu \varepsilon C}}{{{\omega _{r{d_f}}}}}} \right)^\mu }\frac{{{\mu ^\mu }\delta }}{{\Gamma \left( \mu  \right)\omega _{sr}^\mu \mu !}} \propto \frac{1}{{{\rho ^\mu }}},
\end{align}
where $\delta {\rm{ = }}\int_0^\infty  {{x^{ - 1}}{e^{ - \frac{{\mu x}}{{{\omega _{sr}}}}}}} dx$.

Substitute \eqref{appro1 op} and \eqref{appro2 op} into \eqref{expression OP for df}, the asymptotic outage probability for $d_{f}$ can be expression as
\begin{align}\label{asymptotic exoression for the f-th user}
 P_{{d_f}}^\infty  =& \frac{{M!}}{{\left( {M - f} \right)!f!}}{\left( {\frac{{\mu \varepsilon }}{{{\omega _{s{d_f}}}}}} \right)^{\mu f}}{\left( {\frac{1}{{\mu !}}} \right)^f} \nonumber \\
 & \times \left[ {{{\left( {\frac{{\mu \varepsilon }}{{{\omega _{sr}}}}} \right)}^\mu }\left( {\frac{1}{{\mu !}}} \right) + {{\left( {\frac{{\mu \varepsilon C}}{{{\omega _{r{d_f}}}}}} \right)}^\mu }\frac{{{\mu ^\mu }\delta }}{{\Gamma \left( \mu  \right)\omega _{sr}^\mu \mu !}}} \right].
\end{align}
\begin{remark}\label{remark1}
Upon substituting \eqref{asymptotic exoression for the f-th user} into \eqref{diversity order}, the diversity order achieved for $d_{f}$ is $\mu(f+1) $ in the first scenario.
\end{remark}
%it is clear that the diversity order achieved for $d_{f}$ is $\mu(f+1) $, according to cooperative NOMA transmission.
Similar to \eqref{asymptotic exoression for the f-th user}, the asymptotic outage probability for $d_{n}$ can be expression as
\begin{align}\label{asymptotic exoression for the n-th user}
 P_{{d_n}}^\infty  =& \frac{{M!}}{{\left( {M - {n}} \right)!{n}!}}{\left( {\frac{{\mu \Omega }}{{{\omega _{s{d_n}}}}}} \right)^{\mu {n}}}{\left( {\frac{1}{{\mu !}}} \right)^{n}} \nonumber \\
  &\times \left[ {{{\left( {\frac{{\mu \Omega }}{{{\omega _{sr}}}}} \right)}^\mu }\left( {\frac{1}{{\mu !}}} \right) + {{\left( {\frac{{\mu \Omega C}}{{{\omega _{r{d_n}}}}}} \right)}^\mu }\frac{{{\mu ^\mu }\delta }}{{\Gamma \left( \mu  \right)\omega _{sr}^\mu \mu !}}} \right].
\end{align}
\begin{remark}\label{remark2}
Upon substituting \eqref{asymptotic exoression for the n-th user} into \eqref{diversity order}, the diversity order achieved for $d_{n}$ is $\mu(n+1)$ in the first scenario.
\end{remark}
\textbf{Remark \ref{remark1}} and \textbf{Remark \ref{remark2}} provide insightful guidelines for exploiting the direct link between the BS and users over more general fading channels. The diversity order of the user is relevant to the parameter $\mu$.

In the second scenario, Substituting \eqref{sorted channel gain CDF} into \eqref{expression for second scenario}, the asymptotic outage probability for the $m$-th can be expression as
\begin{align}\label{asymptotic exoression of the m-th user for no relay}
P_m^\infty  = \frac{{M!}}{{\left( {M - m} \right)!m!}}{\left( {\frac{{\mu \varphi _m^*}}{{{\omega _{s{d_m}}}}}} \right)^{\mu m}}{\left( {\frac{1}{{\mu !}}} \right)^m} \propto \frac{1}{{{\rho ^{\mu m}}}}.
\end{align}
\begin{remark}
Similar to the first scenario, Upon substituting \eqref{asymptotic exoression of the m-th user for no relay} into \eqref{diversity order}, the diversity order achieved by the $m$-th user is $\mu m$ in the second scenario.
\end{remark}
\subsection{Throughput Analysis}\label{Throughput Analysis}
In this section, the delay-limited transmission mode is considered for two scenarios over Nakagami-$m$ fading channels. The BS sends information at a constant rate and the system throughput is subjective to the effect of outage probability. It is important to investigate the system throughput in the delay-limited mode for practical implementations. Therefore, the system throughput in the first scenario is expressed as
\begin{align}\label{Throughput Analysis for the first scenario}
{R_{fir}} = \left( {1 - {P_{{d_f}}}} \right){R_f} + \left( {1 - {P_{{d_n}}}} \right){R_n},
\end{align}
where $P_{d_{f}}$ and $P_{d_{n}}$ can be obtained from \eqref{expression OP for df} and \eqref{the last expression OP for dn}, respectively.

Additionally, based on the analytical results for the outage probability in the second scenario, the system throughput with the constant rates is expressed as
\begin{align}\label{Throughput Analysis for the second scenario}
{R_{sec }} = \mathop \sum \limits_{i = 1}^M \left( {1 - {P_i}} \right){R_i},
\end{align}
where ${P_i}$ can be obtained from \eqref{expression for second scenario}.

\section{Numerical Results}
In this section, the numerical results are provided to verify the validity of the derived theoretical expressions for two scenarios over Nakagami-$m$ fading channels. Without loss of the generality, the conventional orthogonal multiple access (OMA) is intended as the benchmark for comparison, where the better user is scheduled. The target rate ${R_0}$ for the orthogonal user is equal to $\sum\nolimits_{i = 1}^M {{R_i}}$ bit per channel user (BPCU).
\subsection{The first scenario}
%\textcolor[rgb]{0.00,0.00,1.00}{}
In the first scenario, the distance between the BS and users is normalised to unity. Let ${d_{sr}}$ denotes the distance between the BS and fixed gain relaying. The average power ${\omega _{sr}} = \frac{1}{{d_{sr}^\alpha }}$ and ${\omega _{rd}} = \frac{1}{{{{\left( {1 - {d_{sr}}} \right)}^\alpha }}}$ can be attained, where $\alpha $ is pathloss exponent setting to be $\alpha=2$.
The power allocation coefficients are ${a _f} = 0.8$, ${a _n} = 0.2$ for $M = 5$. The target rate for the near user $d_{n}$ and far user $d_{f}$ are assumed to be ${R_n}=1.5$ and ${R_f}=1$ BPCU, respectively. The fixed gain for AF relaying is assumed to be ${\rm{\kappa }} = 0.9$.% Based on the different parameter settings, Nakagami-$m$ fading channel can be reduce to multiple types of channel. For instance, the Gaussian channel $(\mu=\frac{1}{2})$ and Rayleigh fading channel $(\mu=1)$ are the special cases of its.

Fig. \ref{Fig. 2} plots the outage probability of two users versus SNR with $\mu =1$. The exact outage probability curves of two users for NOMA over Nakagami-$m$ fading channels are given by numerical simulation and perfectly match with the theoretical results derived in \eqref{the last expression OP for df} and \eqref{the last expression OP for dn}, respectively. The asymptotic outage probability curves of two users are plotted according to \eqref{asymptotic exoression for the f-th user} and \eqref{asymptotic exoression for the n-th user}, respectively. Obviously, the asymptotic curves well approximate the exact curves in the high SNR. We can observe that NOMA is capable of outperforming OMA in terms of outage probability. Additionally, Fig. \ref{OP_for_first_scenario_2_3} plots the theoretical results of outage probability versus SNR with $\mu=2$ and $\mu=3$. It is observed that the considered cooperative NOMA system has lower outage probability with the parameter $\mu$ increasing. This phenomenon can be explained is that the high SNR slope for outage probability is becoming more larger.

%Additionally, the theoretical results of outage probability versus SNR with $\mu=2$ and $\mu=3$ are presented in Fig.\ref{OP_for_first_scenario_2_3}. It is observed that the considered cooperative NOMA system has lower outage probability with the parameter $\mu$ increasing.

Fig. \ref{Throughput_relay_scenario} plots the system throughput versus SNR in delay-limited transmission mode for the first scenario.
The solid curves represent throughput with different values of $\mu$ which is obtained from \eqref{Throughput Analysis for the first scenario}. The dashed curves represent throughput of conventional OMA. As can be observed from the figure, the higher system throughput can be achieved with increasing the values of $\mu$ at the high SNR. This phenomenon can be explained as that this scenario has the lower outage probability on the condition of the larger values of $\mu$. It is worth noting that NOMA achieve larger system throughput compared to conventional OMA.
%the impact on the system's outage performance is greater.  increasing the values of $\mu$,
\begin{figure}[t!]
\centering
\includegraphics[width=3.4in,  height=2.7in]{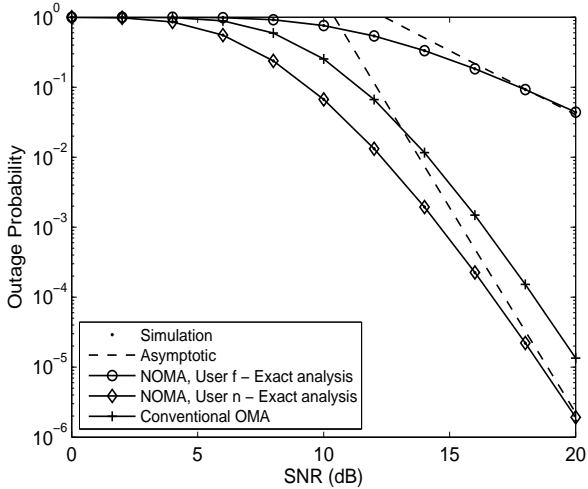}
\caption{ Outage probability for the first scenario versus SNR with $f=1, n=5$ and $\mu$=1.}
\label{Fig. 2}
\end{figure}
\begin{figure}[t!]
\centering
\includegraphics[width=3.3in,  height=2.6in]{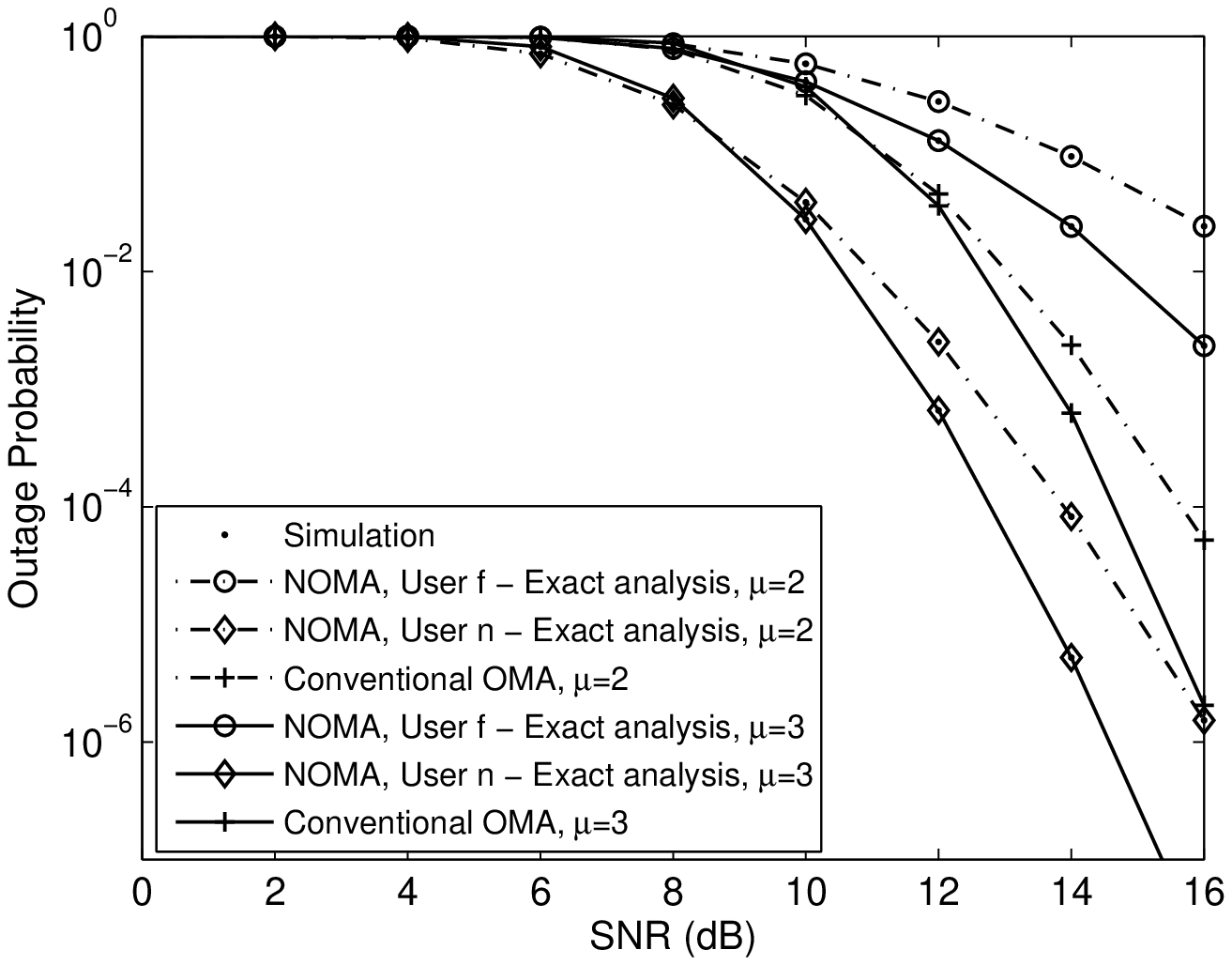}%[width=3.4in,  height=2.7in]
\caption{ Outage probability for the first scenario versus SNR with $f=1, n=5$, $\mu=2$ and $\mu=3$.}
\label{OP_for_first_scenario_2_3}
\end{figure}
\begin{figure}[t!]
\centering
\includegraphics[width=3.3in,  height=2.6in]{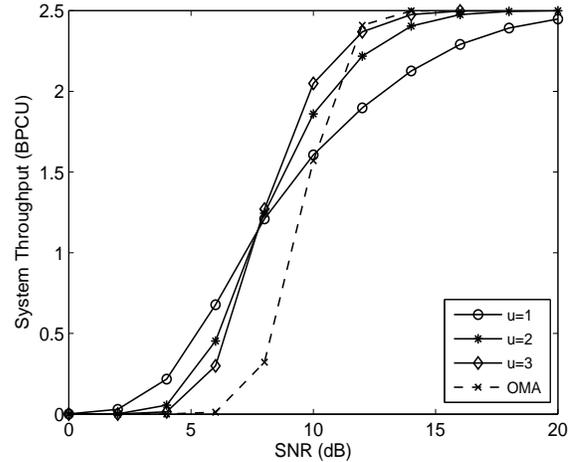}
\caption{System throughput in delay-limited transmission mode versus SNR with $\mu$=1, 2 and 3 for the first scenario.}
\label{Throughput_relay_scenario}
\end{figure}

\subsection{The second scenario}
In the second scenario, we assume that there are three users considered setting to be $M=3$. The average powers between the BS and three users are ${\omega _{1}= 0.3}$, ${\omega _{2}= 1.5}$ and ${\omega _{3}= 5}$, respectively.
The power allocation coefficients are ${a_1} = 0.5$, ${a_2} = 0.4$ and ${a_3} = 0.1$. The target rate for each user is assumed to be ${R_1} = 0.2$, ${R_2} = 1$, ${R_3} = 2$ BPCU, respectively. Similarly, the fixed gain for the AF relaying is also assumed to be ${\rm{\kappa }} = 0.9$.

Fig. \ref{Fig. 4} plots the outage probability of three users versus SNR with $\mu=1$. The solid curves represent the outage probability of three users for NOMA which are obtained from \eqref{op derived for second}. Obviously, the exact outage probability curves match precisely with the Monte Carlo simulation results. The dashed curves represents the asymptotic outage probability which is obtained from \eqref{asymptotic exoression of the m-th user for no relay}. The asymptotic curves well approximate the exact performance curves in the high SNR. It is shown that NOMA is also capable of outperforming orthogonal multiple access (OMA) in terms of outage probability in this scenario. Another observation is that when several users' QoS are met at same time, NOMA scheme offers better fairness with regard to conventional OMA. It is worth pointing out that NOMA and OMA has the same outage probability slope for user 3, which means that they achieves the same diversity. However, the different diversity orders are obtained for user 1 and 2, respectively. Fig. \ref{OP_for_seconde_scenario_2_3} plots the theoretical results of outage probability versus SNR with $\mu=2$ and $\mu=3$. It is worth noting that NOMA system can achieve lower outage performance with the parameter $\mu$ increasing.
%\textcolor[rgb]{0.00,0.00,1.00}{}
The reason is that a larger $\mu$ results in higher diversity order for each user, which in turn leads lower outage probability.
\begin{figure}[t!]
\centering
\includegraphics[width=3.3in,  height=2.6in]{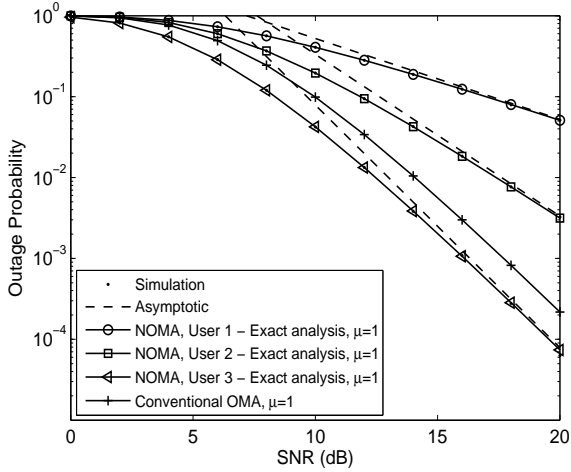}
\caption{ Outage probability for the second scenario versus SNR with $\mu$=1.}
\label{Fig. 4}
\end{figure}
\begin{figure}[t!]
\centering
\includegraphics[width=3.3in,  height=2.6in]{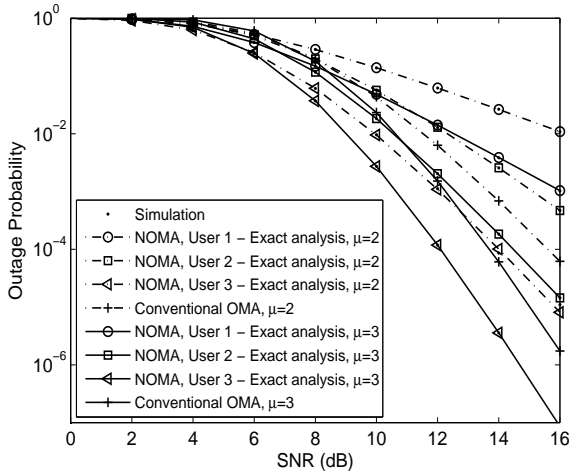}
\caption{ Outage probability for the second scenario versus SNR with $\mu=2$ and 3.}
\label{OP_for_seconde_scenario_2_3}
\end{figure}

Fig. \ref{Throughput_for_seconde_scenario_2_3} plots the system throughput versus SNR in delay-limited transmission mode for the second scenario. The solid curves represent throughput which is obtained from \eqref{Throughput Analysis for the second scenario} with different values of $\mu$. Similarly, one can observe that the higher system throughput can be achieved with the values $\mu$ increasing at the high SNR. As can be seen from the figure, the throughout ceiling exits in the high SNR region. This is due to the fact that the outage probability is tending to zero and throughput is determined only by the target rate.

%\textcolor[rgb]{0.00,0.00,1.00}{ }
From the above analysis results, we observe that the second scenario can be regarded as a benchmark of cooperative NOMA
scenario considered in this paper. For the purposes of comparison, two pairing users (user 1 and user 3) are selected to
perform NOMA jointly. The power allocation coefficients for user 1 and user 3 are ${a_1}=0.8$ and ${a_3}=0.2$.
The target rate for user 1 and user 3 are set to be ${R_1} = 0.5$, ${R_3} = 1$ BPCU, respectively.
Fig. 8 plots the outage probability of two scenarios versus SNR with $\mu=1$. One can observe that the outage performance
of cooperative NOMA scenario is superior to the second scenario. This is due to the fact that cooperative NOMA system can
provide larger diversity order relative to the second scenario.

\begin{figure}[t!]
\centering
\includegraphics[width=3.4in,  height=2.7in]{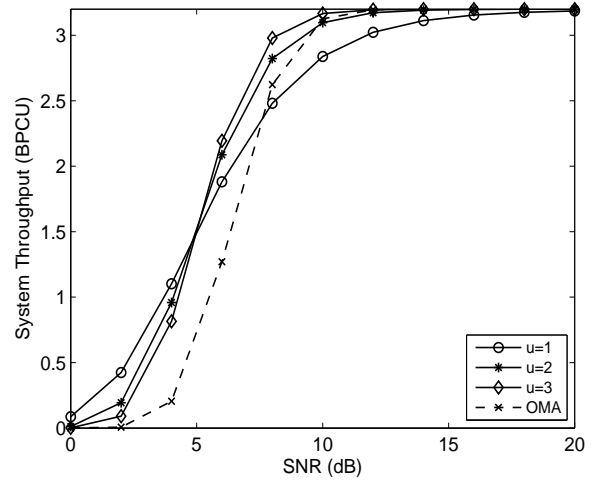}
\caption{ System throughput in delay-limited transmission mode versus SNR with $\mu$=1, 2 and 3 for the second scenario.}
\label{Throughput_for_seconde_scenario_2_3}
\end{figure}
\begin{figure}[t!]
\centering
\includegraphics[width=3.3in,  height=2.6in]{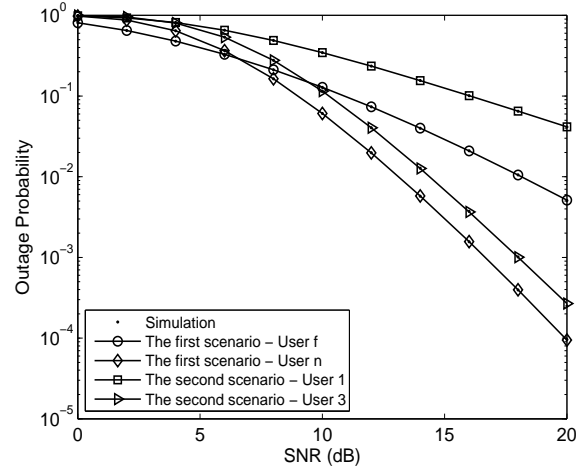}
 \caption{Outage probability for the two scenarios versus SNR with $f=1$, $n=3$ and $\mu$=1.}
\label{scenario_1_vs_2}
\end{figure}

\section{Conclusion}
In this paper the outage performance of NOMA with the fixed gain AF relaying over Nakagami-$m$ fading channels has been investigated. First, the outage behavior of the ordered users by using the AF relaying protocol was researched in detail when the direct links between the BS and the users exist. Second, new closed-form expression for the outage probability with stochastically deployed users was provided under the condition of no relaying node. Based on the analytical results, the diversity orders achieved by the users for the two scenarios have been obtained.
 Furthermore, it is observed that the fairness of multiple users can be ensured by using NOMA scheme in contrast to conventional MA.
 Additionally, these derived results clarified the outage performance of NOMA scheme with cooperative technology over more general fading channels. Finally, the performance of these two scenarios were compared
 in terms of outage probability. Assuming the direct links were existent
in the first scenario, hence our future research may consider comparing the performance between having
direct links and no direct links.
 %\textcolor[rgb]{0.00,0.00,1.00}{It is assumed that these two scenarios are evaluated in terms of outage probability respectively, but our future research may consider comparing the performance of these scenarios.}\textcolor[rgb]{0.00,0.00,1.00}{}

%\begin{center}
%Appendix A Proof of Theorem 1
%\end{center}
\appendices
\section*{Appendix~A: Proof of Theorem \ref{theorem:1}} \label{Appendix:A}
\renewcommand{\theequation}{A.\arabic{equation}}
\setcounter{equation}{0}

Substituting \eqref{far SINR} and \eqref{SINR at df second} into \eqref{expression OP for df}, the outage probability of $d_{f}$ is expressed as follows:\\
\begin{align} \label{OP derived for df}
 {P_{{d_f}}} =& \underbrace {{{\rm{P}}_{\rm{r}}}\left( {\frac{{{{\left| {{h_{s{d_f}}}} \right|}^2}{a _f}\rho }}{{{{\left| {{h_{s{d_f}}}} \right|}^2}{a _n}\rho  + 1}} < {\gamma _{t{h_f}}}} \right)}_{{\Theta _1}} \nonumber \\
  &\times \underbrace {{{\rm{P}}_{\rm{r}}}\left( {\frac{{{{\left| {{h_{sr}}} \right|}^2}{{\left| {{h_{r{d_f}}}} \right|}^2}{a_f}\rho }}{{{{\left| {{h_{sr}}} \right|}^2}{{\left| {{h_{r{d_f}}}} \right|}^2}{a_n}\rho  + {{\left| {{h_{r{d_f}}}} \right|}^2} + C}} < {\gamma _{t{h_f}}}} \right)}_{{\Theta _2}}.
\end{align}

${{{\rm{\Theta }}_1}}$ and ${{{\rm{\Theta }}_1}}$ are calculated as follows:
\begin{align}\label{Theta1 calculation OP for df}
 {{\rm{\Theta }}_1} =& {{\rm{P}}_{\rm{r}}}\left( {{{\left| {{h_{s{d_f}}}} \right|}^2} < \frac{{{\gamma _{t{h_f}}}}}{{\rho \left( {{a _f} - {a _n}{\gamma _{t{h_f}}}} \right)}} \buildrel \Delta \over = \varepsilon } \right) \nonumber\\
  =& \frac{{M!}}{{\left( {f - 1} \right)!\left( {M - f} \right)!}}\mathop \sum \limits_{i = 0}^{M - f} {M-f \choose
  i  }\frac{{{{\left( { - 1} \right)}^i}}}{{f + i}} \nonumber\\
  &\times \mathop \sum \limits_{q = 0}^{f + i} {f+i \choose
  q  }{\left( { - 1} \right)^q}{e^{ - \frac{{\mu \varepsilon q}}{{{\omega _{s{d_f}}}}}}}\nonumber\\
  &\times \mathop \sum \limits_{{p_0} +  \cdots  + {p_{\mu  - 1}} = q} {q \choose
  {{p_0}, \cdots ,{p_{\mu  - 1}}}  }\mathop \prod \limits_{k = 0}^{\mu  - 1} {\left( {\frac{{\psi _f^k}}{{k!}}} \right)^{{p_k}}}, \nonumber\\
\end{align}
where ${{\rm{\Theta }}_1}$ is established on the condition of $\frac{{{a _f}}}{{{a _n}}} > {\gamma _{t{h_f}}}$.
\begin{align}\label{Theta2 calculation OP for df}
 {{\rm{\Theta }}_2} %=& {{\rm{P}}_{\rm{r}}}\left( {\frac{{{{\left| {{h_{sr}}} \right|}^2}{{\left| {{h_{r{d_f}}}} \right|}^2}{\alpha _f}\rho }}{{{{\left| {{h_{sr}}} \right|}^2}{{\left| {{h_{r{d_f}}}} \right|}^2}{\alpha _n}\rho  + {{\left| {{h_{r{d_f}}}} \right|}^2} + C}} < {\gamma _{t{h_f}}}} \right) \nonumber\\
  = &{{\rm{P}}_{\rm{r}}}\left( {{{\left| {{h_{sr}}} \right|}^2} < \varepsilon } \right) \nonumber\\
  &+ {{\rm{P}}_{\rm{r}}}\left( {{{\left| {{h_{r{d_f}}}} \right|}^2} < \frac{{\varepsilon C}}{{\left( {{{\left| {{h_{sr}}} \right|}^2} - \varepsilon } \right)}},{{\left| {{h_{sr}}} \right|}^2} > \varepsilon } \right) \nonumber\\
=& {{\rm{P}}_{\rm{r}}}\left( {{{\left| {{h_{sr}}} \right|}^2} < \varepsilon } \right) \nonumber\\
  &+ \int_\varepsilon ^\infty  {{f_{{{\left| {{h_{sr}}} \right|}^2}}}\left( y \right)} \int_0^{\frac{{\varepsilon C}}{{\left( {y - \varepsilon } \right)}}} {{f_{{{\left| {{h_{r{d_f}}}} \right|}^2}}}\left( x \right)dxdy}  \nonumber\\
 =& 1 - \frac{{{\mu ^\mu }{e^{ - \frac{{\mu \varepsilon }}{{{\omega _{sr}}}}}}}}{{\omega _{sr}^\mu {\rm{\Gamma }}\left( \mu  \right)}}\mathop \sum \limits_{k = 0}^{\mu  - 1} \frac{{{{\left( {\varepsilon C} \right)}^k}}}{{k!}}{\left( {\frac{\mu }{{{\omega _{r{d_f}}}}}} \right)^k}\mathop \sum \limits_{i = 0}^{\mu  - 1} {\mu-1 \choose
  i } \nonumber\\
  &\times {\varepsilon ^{\mu  - i - 1}}\int_0^\infty  {{x^{i - k}}{e^{ - \frac{{\mu \varepsilon C}}{{x{\omega _{r{d_f}}}}} - \frac{{\mu x}}{{{\omega _{sr}}}}}}dx}  \\
   =& 1 - \frac{{2{\mu ^\mu }{e^{ - \frac{{\mu \varepsilon }}{{{\omega _{sr}}}}}}}}{{\omega _{sr}^\mu {\rm{\Gamma }}\left( \mu  \right)}}\mathop \sum \limits_{k = 0}^{\mu  - 1} \frac{{{{\left( {\varepsilon C} \right)}^k}}}{{k!}}{\left( {\frac{\mu }{{{\omega _{r{d_f}}}}}} \right)^k}\mathop \sum \limits_{i = 0}^{\mu  - 1} {\mu-1 \choose
  i } \nonumber\\
  &\times {\varepsilon ^{\mu  - i - 1}}{\left( {\frac{{\varepsilon C{\omega _{{\rm{s}}r}}}}{{{\omega _{r{d_f}}}}}} \right)^{\frac{{i - k + 1}}{2}}}{K_{i - k + 1}}\left( {2\mu \sqrt {\frac{{\varepsilon C}}{{{\omega _{{\rm{s}}r}}{\omega _{r{d_f}}}}}} } \right) , \nonumber\\
\end{align}  %\eqref{Theta2 calculation OP for df}
where \eqref{Theta2 calculation OP for df} follows Binomial theorem and (A.4) is obtained by using [Eq. (3.471.9)] in \cite{gradshteyn}.  Substituting
\eqref{Theta1 calculation OP for df} and (A.4) into \eqref{OP derived for df}, we can obtain
\eqref{the last expression OP for df}.

The proof is completed.
%\begin{center}
%Appendix B Proof of Theorem 2
%\end{center}

\appendices
\section*{Appendix~B: Proof of Theorem \ref{theorem:2}} \label{Appendix:A}
\renewcommand{\theequation}{B.\arabic{equation}}
\setcounter{equation}{0}

Substituting \eqref{near decod far SINR}, \eqref{near SINR} and \eqref{near user detect far user SINR at second}, \eqref{SINR at dn second} into \eqref{expression OP for dn}, the outage probability of $d_{n}$ is expressed below:
\begin{align} \label{OP derived for dn}
 {P_{{d_n}}} = \underbrace {\left[ {1 - {{\rm{P}}_{\rm{r}}}\left( {\frac{{{{\left| {{h_{s{d_n}}}} \right|}^2}{a _f}\rho }}{{{{\left| {{h_{s{d_n}}}} \right|}^2}{a _n}\rho  + 1}} \ge {\gamma _{t{h_f}}},{{\left| {{h_{s{d_n}}}} \right|}^2}{a _n}\rho  \ge {\gamma _{t{h_n}}}} \right)} \right]}_{{\Theta _3}} \nonumber \\
  \times \left[ {1 - {{\rm{P}}_{\rm{r}}}\left( {\frac{{{{\left| {{h_{sr}}} \right|}^2}{{\left| {{h_{r{d_n}}}} \right|}^2}{a _f}\rho }}{{{{\left| {{h_{sr}}} \right|}^2}{{\left| {{h_{r{d_n}}}} \right|}^2}{a _n}\rho  + {{\left| {{h_{r{d_n}}}} \right|}^2} + C}} \ge {\gamma _{t{h_f}}},} \right.} \right.  \nonumber \\
 \underbrace {\begin{array}{*{20}{c}}
   {} & {} & {} & {} & {} & {} & {}  \\
\end{array}\left. {\left. {\frac{{{{\left| {{h_{sr}}} \right|}^2}{{\left| {{h_{r{d_n}}}} \right|}^2}{a _n}\rho }}{{{{\left| {{h_{r{d_n}}}} \right|}^2} + C}} \ge {\gamma _{t{h_n}}}} \right)} \right]}_{{\Theta _4}}
\end{align}

${{\rm{\Theta }}_3}$ and ${{\rm{\Theta }}_4}$ are calculated as follows:
\begin{align}\label{Theta3 calculation OP for dn}
 {{\rm{\Theta }}_3}% =& 1 - {{\rm{P}}_{\rm{r}}}\left( {\frac{{{{\left| {{h_{s{d_n}}}} \right|}^2}{\alpha _f}\rho }}{{{{\left| {{h_{s{d_n}}}} \right|}^2}{\alpha _n}\rho  + 1}} \ge {\gamma _{t{h_f}}}} \right) \nonumber\\
  %&\times {{\rm{P}}_{\rm{r}}}\left( {{{\left| {{h_{s{d_n}}}} \right|}^2}{\alpha _n}\rho  \ge {\gamma _{t{h_n}}}} \right) \nonumber\\
  =& 1 - {{\rm{P}}_{\rm{r}}}\left( {{{\left| {{h_{s{d_n}}}} \right|}^2} \ge \varepsilon } \right){{\rm{P}}_{\rm{r}}}\left( {{{\left| {{h_{s{d_n}}}} \right|}^2} \ge \frac{{{\gamma _{t{h_n}}}}}{{{a _n}\rho }} \buildrel \Delta \over = \beta } \right) \nonumber\\
  =& 1 - {{\rm{P}}_{\rm{r}}}\left( {{{\left| {{h_{s{d_n}}}} \right|}^2} \ge \max \left( {\varepsilon ,\beta } \right) \buildrel \Delta \over = {\rm{\Omega }}} \right) \nonumber\\
  =& \frac{{M!}}{{\left( {n - 1} \right)!\left( {M - n} \right)!}}\mathop \sum \limits_{i = 0}^{M - n}{M-n \choose
  i }\frac{{{{\left( { - 1} \right)}^i}}}{{n + i}} \nonumber\\
  &\times \mathop \sum \limits_{q = 0}^{n + i} {n+i \choose
  q }{\left( { - 1} \right)^q}{e^{ - \frac{{\mu {\rm{\Omega }}q}}{{{\omega _{s{d_n}}}}}}} \nonumber\\
  &\times \mathop \sum \limits_{{p_0} +  \cdots  + {p_{\mu  - 1}} = q} {q \choose
  {{p_0}, \cdots ,{p_{\mu  - 1}}}  }\mathop \prod \limits_{k = 0}^{\mu  - 1}  {\left( {\frac{{\psi _n^k}}{{k!}}} \right)^{{p_k}}},  \nonumber\\
\end{align}
where ${\psi _n} = \frac{{\mu {\rm{\Omega }}}}{{{\omega _{s{d_n}}}}}$. ${\omega _{s{d_n}}}$ denotes the average power of the link between the BS and $d_n$.

\begin{align} \label{Theta4 calculation OP for dn}
 {{\rm{\Theta }}_4} %=& 1 - {{\rm{P}}_{\rm{r}}}\left( {\frac{{{{\left| {{h_{sr}}} \right|}^2}{{\left| {{h_{r{d_n}}}} \right|}^2}{\alpha _f}\rho }}{{{{\left| {{h_{sr}}} \right|}^2}{{\left| {{h_{r{d_n}}}} \right|}^2}{\alpha _n}\rho  + {{\left| {{h_{r{d_n}}}} \right|}^2} + C}} \ge {\gamma _{t{h_f}}}} \right) \nonumber\\
 %& \times {{\rm{P}}_{\rm{r}}}\left( {\frac{{{{\left| {{h_{sr}}} \right|}^2}{{\left| {{h_{r{d_n}}}} \right|}^2}{\alpha _n}\rho }}{{{{\left| {{h_{r{d_n}}}} \right|}^2} + C}} \ge {\gamma _{t{h_n}}}} \right) \nonumber\\
  =& 1 - {{\rm{P}}_{\rm{r}}}\left( {{{\left| {{h_{r{d_n}}}} \right|}^2} \ge \frac{{\varepsilon C}}{{{{\left| {{h_{sr}}} \right|}^2} - \varepsilon }},{{\left| {{h_{sr}}} \right|}^2} \ge \varepsilon } \right) \nonumber\\
  &\times {{\rm{P}}_{\rm{r}}}\left( {{{\left| {{h_{r{d_n}}}} \right|}^2} \ge \frac{{\beta C}}{{{{\left| {{h_{sr}}} \right|}^2} - \beta }},{{\left| {{h_{sr}}} \right|}^2} \ge \beta } \right)  \nonumber\\
 =& 1 - {{\rm{P}}_{\rm{r}}}\left( {{{\left| {{h_{r{d_n}}}} \right|}^2} \ge \frac{{{\rm{\Omega }}C}}{{{{\left| {{h_{sr}}} \right|}^2} - {\rm{\Omega }}}},{{\left| {{h_{sr}}} \right|}^2} \ge {\rm{\Omega }}} \right) \nonumber\\
  =& 1 - \int_{\rm{\Omega }}^\infty  {{f_{{{\left| {{h_{sr}}} \right|}^2}}}\left( y \right)} \int_{\frac{{{\rm{\Omega }}C}}{{y - {\rm{\Omega }}}}}^\infty  {{f_{{{\left| {{h_{r{d_n}}}} \right|}^2}}}\left( x \right)dxdy}  \nonumber\\
  =& 1 - \frac{{{\mu ^\mu }{e^{ - \frac{{\mu {\rm{\Omega }}}}{{{\omega _{sr}}}}}}}}{{\omega _{sr}^\mu {\rm{\Gamma }}\left( \mu  \right)}}\mathop \sum \limits_{k = 0}^{\mu  - 1} \frac{{{{\left( {{\rm{\Omega }}C} \right)}^k}}}{{k!}}{\left( {\frac{\mu }{{{\omega _{r{d_n}}}}}} \right)^k}\mathop \sum \limits_{i = 0}^{\mu  - 1} {\mu-1 \choose
  i } \nonumber\\
  &\times {{\rm{\Omega }}^{\mu  - i - 1}}\int_0^\infty  {{x^{i - k}}{e^{ - \frac{{\mu {\rm{\Omega }}C}}{{x{\omega _{r{d_n}}}}} - \frac{{\mu x}}{{{\omega _{sr}}}}}}dx}  \nonumber\\
  =& 1 - \frac{{2{\mu ^\mu }{e^{ - \frac{{\mu {\rm{\Omega }}}}{{{\omega _{sr}}}}}}}}{{\omega _{sr}^\mu {\rm{\Gamma }}\left( \mu  \right)}}\mathop \sum \limits_{k = 0}^{\mu  - 1} \frac{{{{\left( {{\rm{\Omega }}C} \right)}^k}}}{{k!}}{\left( {\frac{\mu }{{{\omega _{r{d_n}}}}}} \right)^k}\mathop \sum \limits_{i = 0}^{\mu  - 1} {\mu-1 \choose
  i } \nonumber\\
  &\times {{\rm{\Omega }}^{\mu  - i - 1}}{\left( {\frac{{{\rm{\Omega }}C{\omega _{{\rm{sr}}}}}}{{{\omega _{r{d_n}}}}}} \right)^{\frac{{i - k + 1}}{2}}}{K_{i - k + 1}}\left( {2\mu \sqrt {\frac{{{\rm{\Omega }}C}}{{{\omega _{sr}}{\omega _{r{d_n}}}}}} } \right).
\end{align}
Substituting \eqref{Theta3 calculation OP for dn} and \eqref{Theta4 calculation OP for dn} into \eqref{OP derived for dn}, we can obtain \eqref{the last expression OP for dn}.

The proof is completed.

\bibliographystyle{IEEEtran}
\bibliography{mybib}

\end{document}